\documentclass[review]{elsarticle}
\usepackage{mathrsfs}
\usepackage{amssymb}
\usepackage{geometry}
\usepackage{amsmath}

\allowdisplaybreaks[4]

\usepackage{amsfonts}

\usepackage{amssymb,amsmath,amsthm,latexsym}
\usepackage[colorlinks,linkcolor=blue,anchorcolor=blue,citecolor=blue,urlcolor=blue]{hyperref}
\usepackage{amsmath}
\usepackage{lineno,hyperref}
\modulolinenumbers[5]

\journal{Journal of \LaTeX\ Templates}

\newtheorem{thm}{Theorem}[section]
\newtheorem{cor}[thm]{Corollary}
\newtheorem{lem}[thm]{Lemma}

\newtheorem{exa}[thm]{Example}

\numberwithin{equation}{section}
\newcommand{\ord}{\operatorname{ord}}

\newcommand{\Tr}{\operatorname{Tr}}

\begin{document}

\begin{frontmatter}

\title{A construction of $q$-ary linear codes with two weights\tnoteref{mytitlenote}}
\tnotetext[mytitlenote]{The paper is supported by Foundation of Science and Technology on Information Assurance Laboratory (No. KJ-15-009).}

\author[mymainaddress,mysecondaryaddress,mythirdlyaddress]{Ziling Heng}
\ead{zilingheng@163.com}
\author[mymainaddress,mysecondaryaddress,mythirdlyaddress]{Qin Yue}
\ead{yueqin@nuaa.edu.cn}

\address[mymainaddress]{Department of Mathematics, Nanjing University of Aeronautics and Astronautics,
Nanjing, 211106, PR China}
\address[mysecondaryaddress]{State Key Laboratory of Cryptology, P. O. Box 5159, Beijing, 100878, PR China}
\address[mythirdlyaddress]{State Key Laboratory of Information Security, Institute of Information Engineering, Chinese
Academy of Sciences, Beijing, 100093, PR China}
\begin{abstract}
Linear codes with a few weights are very important in coding theory and have attracted a lot of attention. In this paper, we present a construction of $q$-ary linear codes from trace and norm functions over finite fields. The weight distributions of the linear codes are determined in some cases based on Gauss sums. It is interesting that our construction can produce optimal or almost optimal codes. Furthermore, we show that our codes can be used to construct secret sharing schemes with interesting access structures and strongly regular graphs with new parameters.
\end{abstract}


\end{frontmatter}

\linenumbers

\section{Introduction}
Let $\Bbb F_{q}$ denote the finite field with $q$ elements. An $[n,k,d]$ linear code $\mathcal{C}$ over $\Bbb F_{q}$ is a $k$-dimensional subspace of $\Bbb F_{q}^{n}$ with minimum Hamming distance $d$. An $[n,k,d]$ code is called optimal if no $[n,k,d+1]$ code exists. Let $A_{i}$ denote the number of codewords with Hamming weight $i$ in a code $\mathcal{C}$ with length $n$. The weight enumerator of $\mathcal{C}$ is defined by $1+A_{1}z+\cdots+A_{n}z^{n}$. The sequence $(1,A_1,\cdots,A_{n})$ is called the weight distribution of $\mathcal{C}$. The code $\mathcal{C}$ is said to be $t$-weight if the number of nonzero $A_j,1\leq j \leq n$, in the sequence $(1,A_1,\cdots,A_{n})$ equals $t$. The weight distribution is an interesting topic which was investigated in \cite{BM, D, DD, LY1, LY2, ZD} and many other papers. In particular, a survey of three-weight cyclic codes and their weight distributions were provided in \cite{DLLZ}. Weight distribution gives the minimum distance and the error correcting capability of a code. In addition, it contains important information on the computation of the probability of error detection and correction with respect to some error detection and correction algorithms \cite{K1}.

Recently, Ding \emph{et al.} proposed a very effective construction of linear codes in \cite{D, DD} as follows.  Let $D=\{d_1,d_2,\cdots,d_n\}\subseteq \Bbb F_{r}$, where $r$ is a power of $q$. A linear code of length $n$ over $\Bbb F_{q}$ is defined by
$$\mathcal{C}_{D}=\left\{(\Tr_{r/q}(xd_{1}),\Tr_{r/q}(xd_{2}),\cdots,\Tr_{r/q}(xd_{n})):x\in \Bbb F_{r}\right\},$$
where $\Tr_{r/q}(x)=x+x^{q}+\cdots+x^{q^{s-1}}$ denotes the trace function from $\Bbb F_{r}$ to $\Bbb F_{q}$ and $r=q^{s}$. The set $D$ is called the defining set of $\mathcal{C}$. If the set $D$ is well chosen, the code $\mathcal{C}$ may have good parameters. By using this construction and selecting proper defining sets, many good codes were found in \cite{D, DD, HY, HY0, XCX, YY, ZF}. Let $f$ be a function over $\Bbb F_{r}$. Then this construction can be equivalently written as
$$\mathcal{C}_{D}=\left\{(\Tr_{r/q}(xf(d_{1})),\Tr_{r/q}(xf(d_{2})),\cdots,\Tr_{r/q}(xf(d_{n}))):x\in \Bbb F_{r}\right\}.$$

Let $m,m_{1},m_{2}$ be positive integers such that $m_{1}|m,m_{2}|m$ and $\gcd(m_{1},m_{2})=e$. Let $\Tr_{q^{m_{i}}/q}$ be the trace function from $\Bbb F_{q^{m_{i}}}$ to $\Bbb F_{q}$, $i=1,2$. Let $\mathrm{N}_{q^{m}/q^{m_{i}}}$ be the norm function from $\Bbb F_{q^{m}}$ to $\Bbb F_{q^{m_{i}}}$, then for $x\in \Bbb F_{q^m}$,
$$\mathrm{N}_{q^{m}/q^{m_{i}}}(x)=x^{1+q^{m_{i}}+\cdots+(q^{m_{i}})^{\frac{m}{m_{i}}-1}}=x^{\frac{q^{m}-1}{q^{m_{i}}-1}}, i=1,2.$$
In this paper, we present a construction of a $q$-ary linear code as
\begin{eqnarray}\mathcal{C}_{D}=\left\{\left(\Tr_{q^{m_{1}}/q}(x\mathrm{N}_{q^{m}/q^{m_{1}}}(d_{1})),\cdots,\Tr_{q^{m_{1}}/q}
(x\mathrm{N}_{q^{m}/q^{m_{1}}}(d_{n}))\right):x\in \Bbb F_{q^{m_{1}}}\right\},\end{eqnarray}
where the defining set $D$ is given as
$$D=\{x\in \Bbb F_{q^{m}}^{*}:\Tr_{q^{m_{2}}/q}(\mathrm{N}_{q^{m}/q^{m_{2}}}(x))+a=0\}$$ for $a\in \Bbb F_{q}$.
Since the norm function $\mathrm{N}_{q^{m}/q^{m_{2}}}: \Bbb F_{q^m}^*\rightarrow \Bbb F_{q^{m_2}}^*$ is surjective, there exists an element $c\in \Bbb F_{q^{m}}^{*}$ such that $\mathrm{N}_{q^{m}/q^{m_{2}}}(c)=\frac{1}{a}$ for $a\in \Bbb F_{q}^{*}$. If $a\in \Bbb F_{q}^{*}$, then
\begin{eqnarray*}
D&=&\{x\in \Bbb F_{q^{m}}^{*}:\Tr_{q^{m_{2}}/q}(\mathrm{N}_{q^{m}/q^{m_{2}}}(x))+a=0\}\\
&=&\{x\in \Bbb F_{q^{m}}^{*}:\Tr_{q^{m_{2}}/q}(\mathrm{N}_{q^{m}/q^{m_{2}}}(cx))+1=0\}\\
&=&\{x\in \Bbb F_{q^{m}}^{*}:\Tr_{q^{m_{2}}/q}(\mathrm{N}_{q^{m}/q^{m_{2}}}(x))+1=0\}.
\end{eqnarray*}
This implies that we only need to consider $a=0,1$. We remark that this construction is a further generalization of that in \cite{HY1}. When $m=m_1$, in \cite{HY1}, the authors determined a lower bound of the minimum Hamming distance of $\mathcal{C}_{D}$ and gave its weight distributions for $m_2=2,a=0$ and $m_2=1,2,a=1$, respectively.

The purpose of this paper is to determine the weight distribution of $\mathcal{C}_{D}$ defined in Equation (1.1) in some cases. Our main mathematical tools used in this paper are Gauss sums. Consequently, we obtain four classes of two-weight linear codes with very flexible parameters. Examples given by us show that some codes are optimal or almost optimal. As some applications, our codes are used to construct secret sharing schemes with interesting access structures and strongly regular graphs with new parameters.

The following notations will be used in this paper:

\begin{tabular}{ll}
$\chi$, $\chi_1$, $\chi_2$ & canonical additive characters of $\Bbb F_{q}$, $\Bbb F_{q^{m_1}}$, $\Bbb F_{q^{m_2}}$, respectively;\\
$\lambda$, $\lambda_1$, $\lambda_2$ & generators of  multiplicative character groups of $\Bbb F_{q}$, $\Bbb F_{q^{m_1}}$, $\Bbb F_{q^{m_2}}$, respectively;\\
$G(\lambda)$, $G(\lambda_1)$, $G(\lambda_2)$ & Gauss sums over $\Bbb F_{q}$, $\Bbb F_{q^{m_1}}$, $\Bbb F_{q^{m_2}}$, respectively;\\
$\alpha$ & primitive element of $\Bbb F_{q^{m}}^{*}$;\\
$e$ & $e=\gcd(m_1,m_2)$;\\
$l$ & $l=\gcd(\frac{m_1}{e},q-1)$.
\end{tabular}

\section{Gauss sums}

In this section, we recall some basic results of Gauss sums which are important tools in this paper.

Let $\Bbb F_{q}$ be a finite field with $q$ elements, where $q$ is a power of a prime $p$. The canonical additive character of $\Bbb F_{q}$ is defined as follows:
$$\chi: \Bbb F_{q}\longrightarrow \Bbb C^{*}, \chi(x)=\zeta_{p}^{\Tr_{q/p}(x)},$$
where $\zeta_{p}=e^{\frac{2\pi\sqrt{-1}}p}$ denotes the $p$-th primitive root of complex  unity and $\Tr_{q/p}$ is the trace function from $\Bbb F_{q}$ to $\Bbb F_{p}$. The orthogonal property of additive characters (see \cite{L}) is given by:
$$ \sum_{x\in \Bbb F_{q}}\chi(ax)=\left\{
\begin{array}{ll}
  q,   &      \mbox{if}\ a=0,\\
0 & \mbox{otherwise}.
\end{array} \right. $$
Let $\psi: \Bbb F_{q}^{*}\longrightarrow \Bbb C^{*}$ be a multiplicative character of $\Bbb F_{q}^{*}$. For $\Bbb F_{q}^{*}=\langle\beta\rangle$, $\psi(\beta)=\zeta_{q-1}^{i}$ for some $0\leq i \leq q-2$. The trivial multiplicative character $\psi_{0}$ is defined by $\psi_{0}(x)=1$ for all $x\in \Bbb F_{q}^{*}$.  It is known from \cite{L} that all the multiplicative characters form a multiplication group $\widehat{\Bbb F}_{q}^{*}$, which is isomorphic to $\Bbb F_{q}^{*}$. The orthogonal property of a multiplicative character $\psi$  (see \cite{L}) is given by:
$$ \sum_{x\in \Bbb F_{q}^{*}}\psi(x)=\left\{
\begin{array}{ll}
  q-1,   &      \mbox{if}\ \psi=\psi_{0},\\
0,& \mbox{otherwise}.
\end{array} \right. $$

The \emph{Gauss sum} over $\Bbb F_{q}$ is defined by
$$G(\psi)=\sum_{x\in \Bbb F_{q}^{*}}\psi(x)\chi(x).$$
It is easy to see that $G(\psi_{0})=-1$ and $G(\bar\psi)=\psi(-1)\overline{G(\psi)}$. If $\psi \neq \psi_{0}$, we have $|G(\psi)|=\sqrt{q}$.  Gauss sums can be viewed as the Fourier coefficients in the Fourier expansion of the restriction of  $\chi$ to $\Bbb F_{q}^{*}$ in terms of the multiplicative characters of $\Bbb F_{q}$, i.e.
\begin{eqnarray}\chi(x)=\frac{1}{q-1}\sum_{\psi\in \widehat{\Bbb F}_{q}^{*}}G(\bar{\psi})\psi(x),x\in \Bbb F_{q}^{*}.\end{eqnarray}
In this paper, Gauss sum is an important tool to compute exponential sums. In general, the explicit determination of Gauss sums is a difficult problem. In some cases, Gauss sums are explicitly determined in \cite{BEW, L}.

In the following, we state the Gauss sums in the so-called semi-primitive case.
\begin{lem} \cite[Semi-primitive case Gauss sums]{BEW}
Let $\phi$ be a multiplicative character of order $N$ of $\Bbb F_{r}^{*}$. Assume that $N\neq 2$ and there exists a least positive integer $j$ such that $p^{j}\equiv -1\pmod{N}$. Let $r=p^{2j\gamma}$ for some integer $\gamma$. Then the Gauss sums of order $N$ over $\Bbb F_{r}$ are given by
\begin{eqnarray*}G(\phi)=\left\{
\begin{array}{ll}
(-1)^{\gamma-1}\sqrt{r}, & \mbox{if}\ p=2,\\
(-1)^{\gamma-1+\frac{\gamma(p^{j}+1)}{N}}\sqrt{r}, & \mbox{if}\ p\geq3.
\end{array} \right.\end{eqnarray*}
Furthermore, for $1\leq s \leq N-1$, the Gauss sums $G(\phi^{s})$ are given by
\begin{eqnarray*}G(\phi^{s})=\left\{
\begin{array}{ll}
(-1)^{s}\sqrt{r}, & \mbox{if }N\mbox{ is even, } p,\ \gamma\mbox{ and }\frac{p^{j}+1}{N}\mbox{ are odd},\\
(-1)^{\gamma-1}\sqrt{r}, & \mbox{otherwise}.
\end{array} \right.\end{eqnarray*}
\end{lem}

The well-known quadratic Gauss sums are the following.
\begin{lem}\cite[Theorem 5.15]{L}
Suppose that $q=p^{t}$ and $\eta$ is the quadratic multiplicative character of $\Bbb F_{q}$, where $p$ is an odd prime. Then
$$G(\eta)=(-1)^{t-1}(\sqrt{p^{*}})^{t}=\left\{
\begin{array}{ll}
(-1)^{t-1}\sqrt{q}, & \mbox{if}\ p\equiv 1\pmod{4},\\
(-1)^{t-1}(\sqrt{-1})^{t}\sqrt{q}, & \mbox{if}\ p\equiv 3\pmod{4},\\
\end{array}
\right.$$ where $p^{*}=(-1)^{\frac{p-1}{2}}p$.
\end{lem}

\section{Exponential sums}
In this section, we investigate two exponential sums which will be used to calculate the weight distribution of $\mathcal{C}_{D}$.

Let $\chi$ be the canonical additive character of $\Bbb F_{q}$. Let $\chi_{i}$ be the canonical additive character of $\Bbb F_{q^{m_i}}$, $i=1,2$, respectively. Denote
$$\Omega(b)=\sum_{x\in \Bbb F_{q^{m}}^{*}}\sum_{y,z\in \Bbb F_{q}^{*}}\chi_{1}(ybx^{\frac{q^{m}-1}{q^{m_{1}-1}}})\chi_{2}(zx^{\frac{q^{m}-1}{q^{m_{2}}-1}})\chi(z),b\in \Bbb F_{q^{m_{1}}}^{*},$$ and
$$\Delta(b)=\sum_{x\in \Bbb F_{q^{m}}^{*}}\sum_{y,z\in \Bbb F_{q}^{*}}\chi_{1}(ybx^{\frac{q^{m}-1}{q^{m_{1}-1}}})\chi_{2}(zx^{\frac{q^{m}-1}{q^{m_{2}}-1}}),b\in \Bbb F_{q^{m_{1}}}^{*}.$$

Firstly, we begin to compute the exponential sum $\Omega(b)$.
\begin{lem}
Let $m,m_1,m_2$ be positive integers such that $m_{1}|m,m_2|m$ and $\gcd(\frac{m_1}{e},q-1)=l$, where $e=\gcd(m_1,m_2)$. Let $\widehat{\Bbb F}_{q}^{*}=\langle\lambda\rangle$, $\widehat{\Bbb F}_{q^{m_i}}^{*}=\langle\lambda_i\rangle$, $t_i=\frac{q^{m_i}-1}{q^{e}-1}$, $i=1,2$. For $b\in \Bbb F_{q^{m_1}}^{*}$, we have
$$\Omega(b)=\frac{(q^{m}-1)(q-1)}{(q^{m_{1}}-1)(q^{m_{2}}-1)}\sum_{s\in S}G(\bar{\lambda}_{1}^{t_1 s})G(\lambda_{2}^{t_2 s})\lambda_{1}^{t_1 s}(b)G(\bar{\lambda}^{\frac{m_2}{e} s}),$$ where $S=\{\frac{q-1}{l}j:j=0,1,\ldots,\frac{q^{e}-1}{q-1}l-1\}$.
\end{lem}

\begin{proof}
For $\Bbb F_{q^{m}}^{*}=\langle \alpha\rangle$, let $\alpha_{1}=\alpha^{\frac{q^{m}-1}{q^{m_{1}}-1}}$ and $\alpha_{2}=\alpha^{\frac{q^{m}-1}{q^{m_{2}}-1}}$. Then we have $\Bbb F_{q^{m_1}}^{*}=\langle \alpha_{1}\rangle$ and $\Bbb F_{q^{m_2}}^{*}=\langle \alpha_{2}\rangle$. This implies that
\begin{eqnarray*}
\Omega(b)&=&\sum_{y,z\in \Bbb F_{q}^{*}}\sum_{i=0}^{q^{m}-2}\chi_{1}(yb\alpha^{\frac{q^{m}-1}{q^{m_{1}-1}}i})\chi_{2}(z\alpha^{\frac{q^{m}-1}{q^{m_{2}}-1}i})\chi(z)\\
&=&\sum_{y,z\in \Bbb F_{q}^{*}}\sum_{i=0}^{q^{m}-2}\chi_{1}(yb\alpha_{1}^{i})\chi_2(z\alpha_2^{i})\chi(z).
\end{eqnarray*}
Using the Fourier expansion of additive characters (see Equation (2.1)), we have
\begin{eqnarray*}
\Omega(b)&=&\frac{1}{(q^{m_{1}}-1)(q^{m_{2}}-1)}\sum_{y,z\in \Bbb F_{q}^{*}}\chi(z)\sum_{i=0}^{q^{m}-2}\sum_{\psi_{1}\in \widehat{\Bbb F}_{q^{m_1}}^{*}}G(\bar{\psi}_{1})\psi_{1}(yb\alpha_{1}^{i})\sum_{\psi_{2}\in \widehat{\Bbb F}_{q^{m_2}}^{*}}G(\bar{\psi}_{2})\psi_{2}(z\alpha_{2}^{i})\\
&=&\frac{1}{(q^{m_{1}}-1)(q^{m_{2}}-1)}\sum_{y,z\in \Bbb F_{q}^{*}}\chi(z)\sum\limits_{\psi_{j}\in \widehat{\Bbb F}_{q^{m_j}}^{*}\atop j=1,2}G(\bar{\psi}_{1})G(\bar{\psi}_{2})\psi_{1}(yb)\psi_2(z)\sum_{i=0}^{q^{m}-2}\psi_{1}(\alpha_{1}^{i})\psi_{2}(\alpha_{2}^{i}).
\end{eqnarray*}
Since $m_{i}|m$, we obtain $\ord(\psi_i)|(q^{m}-1)$, where $i=1,2$. Therefore, we have
$$(\psi_{1}(\alpha_{1})\psi_{2}(\alpha_{2}))^{q^{m}-1}=1$$ and
\begin{eqnarray*} \sum_{i=0}^{q^{m}-2}\psi_{1}(\alpha_{1}^{i})\psi_{2}(\alpha_{2}^{i})
&=&\sum_{i=0}^{q^{m}-2}(\psi_{1}(\alpha_{1})\psi_{2}(\alpha_{2}))^{i}\\
&=&\left\{
\begin{array}{ll}
  q^{m}-1,   &      \mbox{if}\ \psi_{1}(\alpha_{1})\psi_{2}(\alpha_{2})=1,\\
0 & \mbox{otherwise}.
\end{array} \right.
\end{eqnarray*}
Let $\widehat{\Bbb F}_{q^{m_i}}^{*}=\langle \lambda_{i}\rangle$ such that $\lambda_{i}(\alpha_i)=\zeta_{q^{m_i}-1}$, where $i=1,2$. Assume that $\psi_1=\lambda_1^{u}$ and $\psi_2=\lambda_2^{v}$ for $0\leq u \leq q^{m_1}-2$ and $0\leq v \leq q^{m_2}-2$. If $\psi_{1}(\alpha_{1})\psi_{2}(\alpha_{2})=1$, then $\zeta_{q^{m_1}-1}^{u}\zeta_{q^{m_2}-1}^{v}=1$ which is equivalent to
\begin{eqnarray}(q^{m_2}-1)u+(q^{m_1}-1)v\equiv 0\pmod{(q^{m_1}-1)(q^{m_2}-1)}.\end{eqnarray}
This implies that $(q^{m_2}-1)u+(q^{m_1}-1)v\equiv 0\pmod{q^{m_i}-1},i=1,2$.
Therefore, $(q^{m_2}-1)u \equiv0\pmod{q^{m_1}-1}$ and $(q^{m_1}-1)v \equiv0\pmod{q^{m_2}-1}$. It is known that
$$\gcd(q^{m_{1}}-1,q^{m_{2}}-1)=q^{\gcd(m_1,m_2)}-1=q^{e}-1.$$
Then we have $u\equiv 0\pmod{\frac{q^{m_1}-1}{q^{e}-1}}$ and $v\equiv 0\pmod{\frac{q^{m_2}-1}{q^{e}-1}}$. Denote $u=t_{1}s_1$ and $v=t_{2}s_2$ for $0\leq s_1,s_2 \leq q^{e}-2$, where $t_{1}=\frac{q^{m_1}-1}{q^{e}-1}$, $t_{2}=\frac{q^{m_2}-1}{q^{e}-1}$. Substituting $u=t_{1}s_1,v=t_{2}s_2$ into Equation (3.1), we have
$s_1+s_2=q^{e}-1$. Hence,
\begin{eqnarray*}
\Omega(b)&=&\frac{q^{m}-1}{(q^{m_{1}}-1)(q^{m_{2}}-1)}\sum_{y,z\in \Bbb F_{q}^{*}}\chi(z)\sum_{s_1=0}^{q^{e}-2}G(\bar{\lambda}_{1}^{t_1 s_1})G(\lambda_{2}^{t_2 s_1})\lambda_{1}^{t_1 s_1}(yb)\bar{\lambda}_2^{t_2 s_1}(z)\\
&=&\frac{q^{m}-1}{(q^{m_{1}}-1)(q^{m_{2}}-1)}\sum_{s_1=0}^{q^{e}-2}G(\bar{\lambda}_{1}^{t_1 s_1})G(\lambda_{2}^{t_2 s_1})\lambda_{1}^{t_1 s_1}(b)\sum_{z\in \Bbb F_{q}^{*}}\chi(z)\bar{\lambda}_2^{t_2 s_1}(z)\sum_{y\in \Bbb F_{q}^{*}}\lambda_{1}^{t_1 s_1}(y).
\end{eqnarray*}
Assume that $\Bbb F_{q}^{*}=\langle\beta\rangle$, where $\beta=\alpha_{1}^{\frac{q^{m_1}-1}{q-1}}=\alpha_{2}^{\frac{q^{m_2}-1}{q-1}}$. Hence, $\lambda_1(\beta)=\lambda_2(\beta)=\zeta_{q-1}$. Since $\gcd(\frac{m_1}{e},q-1)=l$, we have
\begin{eqnarray*} \sum_{y\in \Bbb F_{q}^{*}}\lambda_{1}^{t_1 s_1}(y)
&=&\sum_{i=0}^{q-2}\lambda_{1}^{t_1 s_1}(\beta^{i})=\sum_{i=0}^{q-2}(\zeta_{q-1}^{\frac{m_1s_1}{e}})^{i}\\
&=&\left\{
\begin{array}{ll}
  q-1,   &      \mbox{if}\ s_1\equiv 0\pmod{\frac{q-1}{l}},\\
0 & \mbox{otherwise}.
\end{array} \right.
\end{eqnarray*}
Denote $S=\{s_{1}:s_1\equiv 0\pmod{\frac{q-1}{l}},0\leq s_1 \leq q^{e}-2\}$. Let $\widehat{\Bbb F}_{q}^{*}=\langle \lambda\rangle$. Since $\lambda_{2}(z)=\lambda(z)$ for $z\in \Bbb F_{q}$, we have that
\begin{eqnarray*}
\Omega(b)
&=&\frac{(q^{m}-1)(q-1)}{(q^{m_{1}}-1)(q^{m_{2}}-1)}\sum_{s_1\in S}G(\bar{\lambda}_{1}^{t_1 s_1})G(\lambda_{2}^{t_2 s_1})\lambda_{1}^{t_1 s_1}(b)\sum_{z\in \Bbb F_{q}^{*}}\chi(z)\bar{\lambda}_2^{t_2 s_1}(z)\\
&=&\frac{(q^{m}-1)(q-1)}{(q^{m_{1}}-1)(q^{m_{2}}-1)}\sum_{s_1\in S}G(\bar{\lambda}_{1}^{t_1 s_1})G(\lambda_{2}^{t_2 s_1})\lambda_{1}^{t_1 s_1}(b)\sum_{z\in \Bbb F_{q}^{*}}\chi(z)\bar{\lambda}^{\frac{m_2}{e} s_1}(z)\\
&=&\frac{(q^{m}-1)(q-1)}{(q^{m_{1}}-1)(q^{m_{2}}-1)}\sum_{s_1\in S}G(\bar{\lambda}_{1}^{t_1 s_1})G(\lambda_{2}^{t_2 s_1})\lambda_{1}^{t_1 s_1}(b)G(\bar{\lambda}^{\frac{m_2}{e} s_1}).
\end{eqnarray*}
The proof is completed.
\end{proof}

We remark that the Fourier expansion of additive characters used in Lemma 3.1 is an effective technique in computing exponential sums. It was also employed in \cite{LY1} to determine the weight distribution of cyclic codes by Li and Yue. By Lemma 3.1, we know that the value distribution of $\Omega(b)$ can be determined if the Gauss sums are known. In the following, we mainly consider some special cases to give the value distribution of $\Omega(b)$.

\begin{lem}
Let $l=1$ and other notations and hypothesises be the same as those of Lemma 3.1. Then the value distribution of $\Omega(b),b\in \Bbb F_{q^{m_1}}^{*}$, is the following.
\begin{itemize}
\item [(1)] If $e=1$, then $$\Omega(b)=\frac{-(q^{m}-1)(q-1)}{(q^{m_{1}}-1)(q^{m_{2}}-1)}, b\in \Bbb F_{q^{m_1}}^{*}.$$
\item [(2)] If $e=2$, then
\begin{eqnarray*}\Omega(b)=\left\{
\begin{array}{ll}
\frac{-(q^{m}-1)(q-1)}{(q^{m_{1}}-1)(q^{m_{2}}-1)}(1+(-1)^{\frac{m_1+m_2}{2}}q^{\frac{m_1+m_2}{2}+1}), & \frac{q^{m_1}-1}{q+1} \mbox{ times}\\
\frac{-(q^{m}-1)(q-1)}{(q^{m_{1}}-1)(q^{m_{2}}-1)}(1+(-1)^{\frac{m_1+m_2}{2}+1}q^{\frac{m_1+m_2}{2}}), & \frac{q(q^{m_1}-1)}{q+1} \mbox{ times}.
\end{array} \right.
\end{eqnarray*}
\end{itemize}
\end{lem}

\begin{proof}
If $l=1$, by Lemma 3.1, we have that $G(\bar{\lambda}^{\frac{m_2}{e} s})=-1$ and
$$\Omega(b)=\frac{-(q^{m}-1)(q-1)}{(q^{m_{1}}-1)(q^{m_{2}}-1)}\sum_{s\in S}G(\bar{\lambda}_{1}^{t_1 s})G(\lambda_{2}^{t_2 s})\lambda_{1}^{t_1 s}(b),$$ where $S=\{(q-1)j:j=0,1,\ldots,\frac{q^{e}-1}{q-1}-1\}$.
In the following, we discuss the value distribution of the exponential sum $\Omega(b)$ for $e=1,2$, respectively.

(1)Assume that $e=1$. It is clear that $S=\{0\}$. Then $$\Omega(b)=\frac{-(q^{m}-1)(q-1)}{(q^{m_{1}}-1)(q^{m_{2}}-1)}, b\in \Bbb F_{q^{m_1}}^{*}.$$

(2) Assume that $e=2$. Then we have $S=\{(q-1)j:j=0,1,\ldots,q\}$. Hence,
\begin{eqnarray*}\Omega(b)&=&\frac{-(q^{m}-1)(q-1)}{(q^{m_{1}}-1)(q^{m_{2}}-1)}\sum_{j=0}^{q}
G(\bar{\lambda}_{1}^{\frac{q^{m_1}-1}{q+1}j})G(\lambda_{2}^{\frac{q^{m_2}-1}{q+1}j})
\lambda_{1}^{\frac{q^{m_1}-1}{q+1}j}(b)\\
&=&\frac{-(q^{m}-1)(q-1)}{(q^{m_{1}}-1)(q^{m_{2}}-1)}(1+
\sum_{j=1}^{q}G(\bar{\lambda}_{1}^{\frac{q^{m_1}-1}{q+1}j})G(\lambda_{2}^{\frac{q^{m_2}-1}{q+1}j})
\lambda_{1}^{\frac{q^{m_1}-1}{q+1}j}(b)).
\end{eqnarray*}
Note that $\ord(\lambda_{2}^{\frac{q^{m_2}-1}{q+1}})=\ord(\lambda_{1}^{\frac{q^{m_1}-1}{q+1}})=q+1$. Now we give the value distribution of $\Omega(b)$ in several cases.
\begin{itemize}
\item If $q$ is even, by Lemma 2.1 we have
$$G(\bar{\lambda}_{1}^{\frac{q^{m_1}-1}{q+1}j})G(\lambda_{2}^{\frac{q^{m_2}-1}{q+1}j})=(-1)^{\frac{m_1+m_2}{2}}q^{\frac{m_1+m_2}{2}},1\leq j \leq q.$$
Then
$$\Omega(b)=\frac{-(q^{m}-1)(q-1)}{(q^{m_{1}}-1)(q^{m_{2}}-1)}(1+(-1)^{\frac{m_1+m_2}{2}}q^{\frac{m_1+m_2}{2}}
\sum_{j=1}^{q}
\lambda_{1}^{\frac{q^{m_1}-1}{q+1}j}(b)).$$
Let $b=\alpha_1^{s},0\leq s \leq q^{m_1}-2$. Then we have
\begin{eqnarray*}\sum_{j=1}^{q}
\lambda_{1}^{\frac{q^{m_1}-1}{q+1}j}(b)= \sum_{j=1}^{q}\zeta_{q+1}^{js}=\left\{
\begin{array}{ll}
q, & \mbox{if}\ s\equiv 0\pmod{q+1}\\
-1, & \mbox{otherwise}.\\
\end{array} \right.
\end{eqnarray*} Hence, the value distribution of $\Omega(b)$ is
\begin{eqnarray*}\Omega(b)=\left\{
\begin{array}{ll}
\frac{-(q^{m}-1)(q-1)}{(q^{m_{1}}-1)(q^{m_{2}}-1)}(1+(-1)^{\frac{m_1+m_2}{2}}q^{\frac{m_1+m_2}{2}+1}), & \frac{q^{m_1}-1}{q+1} \mbox{ times,}\\
\frac{-(q^{m}-1)(q-1)}{(q^{m_{1}}-1)(q^{m_{2}}-1)}(1+(-1)^{\frac{m_1+m_2}{2}+1}q^{\frac{m_1+m_2}{2}}), & \frac{q(q^{m_1}-1)}{q+1} \mbox{ times}.
\end{array} \right.
\end{eqnarray*}
\item If $q$ is odd and $m_2\equiv 0\pmod{4}$, we have $m_1\equiv 2\pmod{4}$ due to $\gcd(\frac{m_1}{2},q-1)=1$. Since $\frac{m_1}{2}$ is odd and $\frac{m_2}{2}$ is even, by Lemma 2.1 we have
    $$G(\bar{\lambda}_{1}^{\frac{q^{m_1}-1}{q+1}j})G(\lambda_{2}^{\frac{q^{m_2}-1}{q+1}j})=(-1)^{j+1}q^{\frac{m_1+m_2}{2}},1\leq j \leq q.$$ For $b=\alpha_1^{s},0\leq s \leq q^{m_1}-2$,
\begin{eqnarray*}\Omega(b)&=&\frac{-(q^{m}-1)(q-1)}{(q^{m_{1}}-1)(q^{m_{2}}-1)}(1+q^{\frac{m_1+m_2}{2}}
\sum_{j=1}^{q}(-1)^{j+1}
\lambda_{1}^{\frac{q^{m_1}-1}{q+1}j}(b))\\
&=&\frac{-(q^{m}-1)(q-1)}{(q^{m_{1}}-1)(q^{m_{2}}-1)}(1+q^{\frac{m_1+m_2}{2}}
\sum_{j=1}^{q}(-1)^{j+1}\zeta_{q+1}^{js}).
\end{eqnarray*}
For $s\equiv 0\pmod{q+1}$, we have
$$\sum_{j=1}^{q}(-1)^{j+1}\zeta_{q+1}^{js}=1\mbox{ and }\Omega(b)=\frac{-(q^{m}-1)(q-1)}{(q^{m_{1}}-1)(q^{m_{2}}-1)}(1+q^{\frac{m_1+m_2}{2}}).$$
For $s\equiv \frac{q+1}{2}\pmod{q+1}$, we have
$$\sum_{j=1}^{q}(-1)^{j+1}\zeta_{q+1}^{js}=-q \mbox{ and }\Omega(b)=\frac{-(q^{m}-1)(q-1)}{(q^{m_{1}}-1)(q^{m_{2}}-1)}(1-q^{\frac{m_1+m_2}{2}+1}).$$
For $s\not\equiv 0,\frac{q+1}{2}\pmod{q+1}$, one can see that
$$\zeta_{q+1}^{s}+\zeta_{q+1}^{3s}+\zeta_{q+1}^{5s}+\cdots+\zeta_{q+1}^{qs}=0$$ and
$$\zeta_{q+1}^{2s}+\zeta_{q+1}^{4s}+\zeta_{q+1}^{6s}+\cdots+\zeta_{q+1}^{(q-1)s}=-1.$$
This implies that $$\sum_{j=1}^{q}(-1)^{j+1}\zeta_{q+1}^{js}=1\mbox{ and }\Omega(b)=\frac{-(q^{m}-1)(q-1)}{(q^{m_{1}}-1)(q^{m_{2}}-1)}(1+q^{\frac{m_1+m_2}{2}}).$$
Hence, the value distribution of $\Omega(b)$ is
\begin{eqnarray*}\Omega(b)=\left\{
\begin{array}{ll}
\frac{(q^{m}-1)(q-1)}{(q^{m_{1}}-1)(q^{m_{2}}-1)}(q^{\frac{m_1+m_2}{2}+1}-1), & \frac{q^{m_1}-1}{q+1} \mbox{ times,}\\
\frac{-(q^{m}-1)(q-1)}{(q^{m_{1}}-1)(q^{m_{2}}-1)}(1+q^{\frac{m_1+m_2}{2}}), & \frac{q(q^{m_1}-1)}{q+1} \mbox{ times}.
\end{array} \right.
\end{eqnarray*}
\item If $q$ is odd and $m_2\equiv 2\pmod{4}$, we have $m_1\equiv 2\pmod{4}$ due to $\gcd(\frac{m_1}{2},q-1)=1$. In this case, the value distribution of $\Omega(b)$ can be obtained in a similar way. We omit the details here. The value distribution of $\Omega(b)$ is given as
    \begin{eqnarray*}\Omega(b)=\left\{
\begin{array}{ll}
\frac{-(q^{m}-1)(q-1)}{(q^{m_{1}}-1)(q^{m_{2}}-1)}(1+q^{\frac{m_1+m_2}{2}+1}), & \frac{q^{m_1}-1}{q+1} \mbox{ times,}\\
\frac{-(q^{m}-1)(q-1)}{(q^{m_{1}}-1)(q^{m_{2}}-1)}(1-q^{\frac{m_1+m_2}{2}}), & \frac{q(q^{m_1}-1)}{q+1} \mbox{ times}.
\end{array} \right.
\end{eqnarray*}
\end{itemize}
Note that the value distribution of $\Omega(b)$ can be represented in a unified form for $e=2$. The proof is completed.
\end{proof}

\begin{lem}
Let $l=2,e=1$and other notations and hypothesises be the same as those of Lemma 3.1. Then the value distribution of $\Omega(b),b\in \Bbb F_{q^{m_1}}^{*}$, is given as follows.
\begin{eqnarray*}\Omega(b)=\left\{
\begin{array}{ll}
\frac{(q^{m}-1)(q-1)}{(q^{m_{1}}-1)(q^{m_{2}}-1)}(-1-q^{\frac{m_1+m_2+1}{2}}), & \frac{q^{m_1}-1}{2} \mbox{ times},\\
\frac{(q^{m}-1)(q-1)}{(q^{m_{1}}-1)(q^{m_{2}}-1)}(-1+q^{\frac{m_1+m_2+1}{2}}), & \frac{q^{m_1}-1}{2} \mbox{ times}.\\
\end{array} \right.
\end{eqnarray*}
\end{lem}

\begin{proof}
Since $l=2,e=1$, by Lemma 3.1 we have that
$$\Omega(b)=\frac{(q^{m}-1)(q-1)}{(q^{m_{1}}-1)(q^{m_{2}}-1)}\sum_{s\in S}G(\bar{\lambda}_{1}^{t_1 s})G(\lambda_{2}^{t_2 s})\lambda_{1}^{t_1 s}(b)G(\bar{\lambda}^{m_2 s}),b\in \Bbb F_{q^{m_1}}^{*},$$ where $S=\{0,\frac{q-1}{2}\}$. It is clear that $m_1$ is even and $m_2$ is odd. Hence, by Lemma 2.2,
\begin{eqnarray*}
\Omega(b)&=&\frac{(q^{m}-1)(q-1)}{(q^{m_{1}}-1)(q^{m_{2}}-1)}(-1+G(\bar{\lambda}_{1}^{\frac{q^{m_1}-1}{2}})G(\lambda_2^{\frac{q^{m_2}-1}{2}})
\lambda_{1}^{\frac{q^{m_1}-1}{2}}(b)G(\bar{\lambda}^{\frac{q-1}{2}m_{2}}))\\
&=&\frac{(q^{m}-1)(q-1)}{(q^{m_{1}}-1)(q^{m_{2}}-1)}(-1+G(\bar{\lambda}_{1}^{\frac{q^{m_1}-1}{2}})G(\lambda_2^{\frac{q^{m_2}-1}{2}})
\lambda_{1}^{\frac{q^{m_1}-1}{2}}(b)G(\bar{\lambda}^{\frac{q-1}{2}}))\\
&=&\frac{(q^{m}-1)(q-1)}{(q^{m_{1}}-1)(q^{m_{2}}-1)}(-1-(-1)^{\frac{(p-1)(m_1+m_2+1)t}{4}}q^{\frac{m_1+m_2+1}{2}}\lambda_{1}^{\frac{q^{m_1}-1}{2}}(b))\\
&=&\left\{
\begin{array}{ll}
\frac{(q^{m}-1)(q-1)}{(q^{m_{1}}-1)(q^{m_{2}}-1)}(-1-q^{\frac{m_1+m_2+1}{2}}), & \frac{q^{m_1}-1}{2} \mbox{ times},\\
\frac{(q^{m}-1)(q-1)}{(q^{m_{1}}-1)(q^{m_{2}}-1)}(-1+q^{\frac{m_1+m_2+1}{2}}), & \frac{q^{m_1}-1}{2} \mbox{ times}.\\
\end{array} \right.
\end{eqnarray*}
\end{proof}
For $l=e=2$, the value distribution of $\Omega(b)$ can't be given because the Gauss sums of order $2(q+1)$ are unknown in general. However, for $e=1$ and $l=3,4$, we can easily obtain the value distributions of $\Omega(b)$ because the cubic and quartic Gauss sums are known. We omit the details here.

In the following, we begin to investigate the exponential sum $\Delta(b),b\in \Bbb F_{q^{m_1}}^{*}$.
\begin{lem}
Let $m,m_1,m_2$ be positive integers such that $m_{1}|m,m_2|m$. Denote $e=\gcd(m_1,m_2)$. Let $t_i=\frac{q^{m_i}-1}{q^{e}-1}$, $i=1,2$. For $b\in \Bbb F_{q^{m_1}}^{*}$, we have
$$\Delta(b)=\frac{(q^{m}-1)(q-1)^{2}}{(q^{m_{1}}-1)(q^{m_{2}}-1)}\sum_{s\in S}G(\bar{\lambda}_{1}^{t_1 s})G(\lambda_{2}^{t_2 s})\lambda_{1}^{t_1 s}(b),$$ where $S=\{(q-1)j:j=0,1,\ldots,\frac{q^{e}-1}{q-1}-1\}$ and $\widehat{\Bbb F}_{q^{m_i}}^{*}=\langle\lambda_i\rangle$, $i=1,2$.
\end{lem}
\begin{proof}
For $\Bbb F_{q^{m}}^{*}=\langle \alpha\rangle$, let $\alpha_{1}=\alpha^{\frac{q^{m}-1}{q^{m_{1}}-1}}$ and $\alpha_{2}=\alpha^{\frac{q^{m}-1}{q^{m_{2}}-1}}$. Then we have $\Bbb F_{q^{m_1}}^{*}=\langle \alpha_{1}\rangle$ and $\Bbb F_{q^{m_2}}^{*}=\langle \alpha_{2}\rangle$. This implies that
\begin{eqnarray*}
\Delta(b)&=&\sum_{y,z\in \Bbb F_{q}^{*}}\sum_{i=0}^{q^{m}-2}\chi_{1}(yb\alpha^{\frac{q^{m}-1}{q^{m_{1}-1}}i})\chi_{2}(z\alpha^{\frac{q^{m}-1}{q^{m_{2}}-1}i})\\
&=&\sum_{y,z\in \Bbb F_{q}^{*}}\sum_{i=0}^{q^{m}-2}\chi_{1}(yb\alpha_{1}^{i})\chi_2(z\alpha_2^{i}).
\end{eqnarray*}
Using the Fourier expansion of additive characters (see Equation (2.1)), we have
\begin{eqnarray*}
\Delta(b)&=&\frac{1}{(q^{m_{1}}-1)(q^{m_{2}}-1)}\sum_{y,z\in \Bbb F_{q}^{*}}\sum_{i=0}^{q^{m}-2}\sum_{\psi_{1}\in \widehat{\Bbb F}_{q^{m_1}}^{*}}G(\bar{\psi}_{1})\psi_{1}(yb\alpha_{1}^{i})\sum_{\psi_{2}\in \widehat{\Bbb F}_{q^{m_2}}^{*}}G(\bar{\psi}_{2})\psi_{2}(z\alpha_{2}^{i})\\
&=&\frac{1}{(q^{m_{1}}-1)(q^{m_{2}}-1)}\sum_{y,z\in \Bbb F_{q}^{*}}\sum\limits_{\psi_{j}\in \widehat{\Bbb F}_{q^{m_j}}^{*}\atop j=1,2}G(\bar{\psi}_{1})G(\bar{\psi}_{2})\psi_{1}(yb)\psi_2(z)\sum_{i=0}^{q^{m}-2}\psi_{1}(\alpha_{1}^{i})\psi_{2}(\alpha_{2}^{i}).
\end{eqnarray*}
Let $t_i=\frac{q^{m_i}-1}{q^{e}-1}$, $i=1,2$. From the proof of Lemma 3.1, we know that
\begin{eqnarray*} \sum_{i=0}^{q^{m}-2}\psi_{1}(\alpha_{1}^{i})\psi_{2}(\alpha_{2}^{i})
=\left\{
\begin{array}{ll}
  q^{m}-1,   &      \mbox{if}\ \psi_{1}(\alpha_{1})\psi_{2}(\alpha_{2})=1,\\
0 & \mbox{otherwise},
\end{array} \right.
\end{eqnarray*}
and $\psi_{1}(\alpha_{1})\psi_{2}(\alpha_{2})=1$ if and only if $\psi_1=\lambda_1^{t_1s_1}$ and $\psi_2=\lambda_2^{t_2s_2}$, where $s_2=q^{e}-1-s_1$ and $0\leq s_1\leq q^{e}-2$. Hence,
\begin{eqnarray*}
\Delta(b)&=&\frac{q^{m}-1}{(q^{m_{1}}-1)(q^{m_{2}}-1)}\sum_{y,z\in \Bbb F_{q}^{*}}\sum_{s_1=0}^{q^{e}-2}G(\bar{\lambda}_{1}^{t_1 s_1})G(\lambda_{2}^{t_2 s_1})\lambda_{1}^{t_1 s_1}(yb)\bar{\lambda}_2^{t_2 s_1}(z)\\
&=&\frac{q^{m}-1}{(q^{m_{1}}-1)(q^{m_{2}}-1)}\sum_{s_1=0}^{q^{e}-2}G(\bar{\lambda}_{1}^{t_1 s_1})G(\lambda_{2}^{t_2 s_1})\lambda_{1}^{t_1 s_1}(b)\sum_{y\in \Bbb F_{q}^{*}}\lambda_{1}^{t_1 s_1}(y)\sum_{z\in \Bbb F_{q}^{*}}\bar{\lambda}_2^{t_2 s_1}(z).
\end{eqnarray*}
Assume that $\Bbb F_{q}^{*}=\langle\beta\rangle$, where $\beta=\alpha_{1}^{\frac{q^{m_1}-1}{q-1}}=\alpha_{2}^{\frac{q^{m_2}-1}{q-1}}$. Hence, $\lambda_1(\beta)=\lambda_2(\beta)=\zeta_{q-1}$. This implies that
\begin{eqnarray*} \sum_{y\in \Bbb F_{q}^{*}}\lambda_{1}^{t_1 s_1}(y)
&=&\sum_{i=0}^{q-2}\lambda_{1}^{t_1 s_1}(\beta^{i})=\sum_{i=0}^{q-2}(\zeta_{q-1}^{\frac{m_1s_1}{e}})^{i}\\
&=&\left\{
\begin{array}{ll}
  q-1,   &      \mbox{if}\ \frac{m_1s_1}{e}\equiv 0\pmod{q-1},\\
0 & \mbox{otherwise},
\end{array} \right.
\end{eqnarray*}
and
\begin{eqnarray*} \sum_{z\in \Bbb F_{q}^{*}}\bar{\lambda}_{2}^{t_1 s_1}(z)
&=&\sum_{i=0}^{q-2}\bar{\lambda}_{2}^{t_2 s_1}(\beta^{i})=\sum_{i=0}^{q-2}(\zeta_{q-1}^{-\frac{m_2s_1}{e}})^{i}\\
&=&\left\{
\begin{array}{ll}
  q-1,   &      \mbox{if}\ \frac{m_2s_1}{e}\equiv 0\pmod{q-1},\\
0 & \mbox{otherwise}.
\end{array} \right.
\end{eqnarray*}
Since $\gcd(\frac{m_1}{e},\frac{m_2}{e})=1$, the system
\begin{eqnarray*} \left\{
\begin{array}{l}
\frac{m_1s_1}{e}\equiv 0\pmod{q-1},\\
\frac{m_2s_1}{e}\equiv 0\pmod{q-1},\\
\end{array} \right.
\end{eqnarray*}
is equivalent to $s_1\equiv 0\pmod{q-1}$, where $0\leq s_1 \leq q^{e}-2$. Denote $S=\{s_{1}:s_1\equiv 0\pmod{q-1},0\leq s_1 \leq q^{e}-2\}$. Then we have that
\begin{eqnarray*}
\Delta(b)=\frac{(q^{m}-1)(q-1)^{2}}{(q^{m_{1}}-1)(q^{m_{2}}-1)}\sum_{s_1\in S}G(\bar{\lambda}_{1}^{t_1 s_1})G(\lambda_{2}^{t_2 s_1})\lambda_{1}^{t_1 s_1}(b).
\end{eqnarray*}
\end{proof}

For $e=1,2$, the value distribution of $\Delta(b)$ can be given as follows.

\begin{lem}
Let the notations be the same as those of Lemma 3.3. Then the value distribution of $\Delta(b),b\in \Bbb F_{q^{m_1}}^{*}$, is given as follows.
\begin{itemize}
\item [(1)] If $e=1$, then $\Delta(b)=\frac{(q^{m}-1)(q-1)^{2}}{(q^{m_{1}}-1)(q^{m_{2}}-1)}$ for all $b\in \Bbb F_{q^{m_1}}^{*}$.
\item [(2)] If $e=2$, then
\begin{eqnarray*}\Delta(b)=\left\{
\begin{array}{ll}
\frac{(q^{m}-1)(q-1)^{2}}{(q^{m_{1}}-1)(q^{m_{2}}-1)}(1+(-1)^{\frac{m_1+m_2}{2}}q^{\frac{m_1+m_2}{2}+1}), & \frac{q^{m_1}-1}{q+1} \mbox{ times}\\
\frac{(q^{m}-1)(q-1)^{2}}{(q^{m_{1}}-1)(q^{m_{2}}-1)}(1+(-1)^{\frac{m_1+m_2}{2}+1}q^{\frac{m_1+m_2}{2}}), & \frac{q(q^{m_1}-1)}{q+1} \mbox{ times}.
\end{array} \right.
\end{eqnarray*}
\end{itemize}
\end{lem}

\begin{proof}
The proof is similar to that of Lemma 3.2. We omit the details here.
\end{proof}

\section{The weight distribution of $\mathcal{C}_{D}$}
In this section, we give the weight distribution of $\mathcal{C}_{D}$ defined in Equation (1.1) in some special cases. The well-known Griesmer bound of linear codes is the following.
\begin{lem}
\cite[Griesmer bound]{MS} For an $[n,k,d]$ code over $\Bbb F_{q}$, we have
$$n\geq \sum_{i=0}^{k-1}\lceil d/q^{i}\rceil,$$
where $\lceil x\rceil$ denotes the smallest integer which is larger than or equal to $x$.
\end{lem}
\subsection{The case $a=0$}
In the following, we determine the weight distribution of $\mathcal{C}_{D}$ for $a=0$.

Denote $n=|D|=|\{x\in \Bbb F_{q^{m}}^{*}:\Tr_{q^{m_{2}}/q}(\mathrm{N}_{q^{m}/q^{m_{2}}}(x))=0\}|$. Since the norm function
$$\mathrm{N}_{q^{m}/q^{m_{2}}}: \Bbb F_{q^{m}}^{*}\longrightarrow \Bbb F_{q^{m_2}}^{*},x\longmapsto x^{\frac{q^{m}-1}{q^{m_2}-1}},$$
is an epimorphism of two multiplicative groups and the trace function $\Tr_{q^{m_{2}}/q}:\Bbb F_{q^{m_2}}\longrightarrow \Bbb F_{q}$ is an epimorphism of two additive groups, we have
\begin{eqnarray}n=|\ker(\mathrm{N}_{q^{m}/q^{m_{2}}})|\cdot (|\ker(\Tr_{q^{m_{2}}/q})|-1)=\frac{q^{m}-1}{q^{m_2}-1}(q^{m_2-1}-1).\end{eqnarray}
Note that $n=0$ when $m_2=1$. Hence, we always assume that $m_2>1$ in this section.
For $b\in \Bbb F_{q^{m_1}}^{*}$, we denote
$$N(b)=|\{x\in \Bbb F_{q^{m}}^{*}:\Tr_{q^{m_{2}}/q}(\mathrm{N}_{q^{m}/q^{m_{2}}}(x))=0\mbox{ and }\Tr_{q^{m_1}/q}(b\mathrm{N}_{q^{m}/q^{m_{1}}}(x))=0\}|.$$
By the basic facts of additive characters, we have that
\begin{eqnarray*}
N(b)&=&\frac{1}{q^{2}}\sum_{x\in \Bbb F_{q^{m}}^{*}}\sum_{y,z\in \Bbb F_{q}}\chi(y\Tr_{q^{m_1}/q}(b\mathrm{N}_{q^{m}/q^{m_{1}}}(x)))\chi(z\Tr_{q^{m_{2}}/q}(\mathrm{N}_{q^{m}/q^{m_{2}}}(x)))\\
&=&\frac{1}{q^{2}}\sum_{x\in \Bbb F_{q^{m}}^{*}}\sum_{y,z\in \Bbb F_{q}}\chi_{1}(ybx^{\frac{q^{m}-1}{q^{m_1}-1}})\chi_2(zx^{\frac{q^{m}-1}{q^{m_2}-1}})\\
&=&\frac{q^{m}-1}{q^{2}}+\frac{1}{q^{2}}\sum_{x\in \Bbb F_{q^{m}}^{*}}\sum_{y\in \Bbb F_{q}^{*}}\chi_{1}(ybx^{\frac{q^{m}-1}{q^{m_1}-1}})+\frac{1}{q^{2}}\sum_{x\in \Bbb F_{q^{m}}^{*}}\sum_{z\in \Bbb F_{q}^{*}}\chi_2(zx^{\frac{q^{m}-1}{q^{m_2}-1}})\\
& &+\frac{1}{q^{2}}\sum_{x\in \Bbb F_{q^{m}}^{*}}\sum_{y,z\in \Bbb F_{q}^{*}}\chi_{1}(ybx^{\frac{q^{m}-1}{q^{m_1}-1}})\chi_2(zx^{\frac{q^{m}-1}{q^{m_2}-1}})\\
&=&\frac{q^{m}-1}{q^{2}}+\frac{1}{q^{2}}\sum_{x\in \Bbb F_{q^{m}}^{*}}\sum_{y\in \Bbb F_{q}^{*}}\chi_{1}(ybx^{\frac{q^{m}-1}{q^{m_1}-1}})+\frac{1}{q^{2}}\sum_{x\in \Bbb F_{q^{m}}^{*}}\sum_{z\in \Bbb F_{q}^{*}}\chi_2(zx^{\frac{q^{m}-1}{q^{m_2}-1}})+\frac{1}{q^{2}}\Delta(b).
\end{eqnarray*}
Note that the norm function $\mathrm{N}_{q^{m}/q^{m_{2}}}$ is an epimorphism. Hence,
\begin{eqnarray*}\sum_{x\in \Bbb F_{q^{m}}^{*}}\sum_{z\in \Bbb F_{q}^{*}}\chi_{2}(zx^{\frac{q^{m}-1}{q^{m_2}-1}})&=&\frac{q^{m}-1}{q^{m_2}-1}\sum_{x\in \Bbb F_{q^{m}}^{*}}\sum_{z\in \Bbb F_{q}^{*}}\chi_{2}(zx)\\
&=&\frac{q^{m}-1}{q^{m_2}-1}\sum_{z\in \Bbb F_{q}^{*}}\sum_{x\in \Bbb F_{q^{m}}^{*}}\chi_{2}(zx)=\frac{(q-1)(1-q^{m})}{q^{m_2}-1}.
\end{eqnarray*}
Similarly,
\begin{eqnarray*}\sum_{x\in \Bbb F_{q^{m}}^{*}}\sum_{y\in \Bbb F_{q}^{*}}\chi_{1}(ybx^{\frac{q^{m}-1}{q^{m_1}-1}})&=&\frac{q^{m}-1}{q^{m_1}-1}\sum_{x\in \Bbb F_{q^{m}}^{*}}\sum_{y\in \Bbb F_{q}^{*}}\chi_{1}(ybx)\\
&=&\frac{q^{m}-1}{q^{m_1}-1}\sum_{y\in \Bbb F_{q}^{*}}\sum_{x\in \Bbb F_{q^{m}}^{*}}\chi_{1}(ybx)=\frac{(q-1)(1-q^{m})}{q^{m_1}-1}.
\end{eqnarray*}
From the discussions above, we obtain that
\begin{eqnarray}
N(b)=\frac{q^{m}-1}{q^{2}}(1-\frac{q-1}{q^{m_2}-1}-\frac{q-1}{q^{m_1}-1})+\frac{1}{q^{2}}\Delta(b).
\end{eqnarray}

For any $b\in \Bbb F_{q^{m_1}}^{*}$, the weight of a codeword
$$\textbf{c}(b)=(\Tr_{q^{m_{1}}/q}(b\mathrm{N}_{q^{m}/q^{m_{1}}}(d_{1})),\cdots,\Tr_{q^{m_{1}}/q}(b\mathrm{N}_{q^{m}/q^{m_{1}}}(d_{n})))$$
equals
\begin{eqnarray}w_{H}(\textbf{c}(b))=n-N(b)=\frac{(q-1)(q^{m}-1)(q^{m_1+m_2}-q^{m_1+1}+q-1)}{q^{2}(q^{m_1}-1)(q^{m_2}-1)}-\frac{1}{q^{2}}\Delta(b)
\end{eqnarray}
by Equations (4.1) and (4.2). Hence, by Lemma 3.5, the parameters of $\mathcal{C}_{D}$ for $e=1$ are
$$[\frac{q^{m}-1}{q^{m_2}-1}(q^{m_2-1}-1),m_1,\frac{q^{m_1-1}(q-1)(q^{m}-1)(q^{m_2-1}-1)}{(q^{m_1}-1)(q^{m_2}-1)}].$$
Then $\mathcal{C}_{D}$ is an optimal one-weight linear code with respect to the Griesmer bound. However, any one-weight linear code is not new because it is equivalent to a concatenated version of a simplex code. For $e=2$, the weight distribution of $\mathcal{C}_{D}$ is given in the following.

\begin{thm}
Let $m,m_1,m_2$ be positive integers such that $m_1|m$, $m_2|m$, and $m_2>1$. Denote $\gcd(m_1,m_2)=e$. Let $\mathcal{C}_{D}$ be the linear code defined in Equation (1.1) for $a=0$. If $e=2$ and $(m_1,m_2)\neq (2,2)$, then $\mathcal{C}_{D}$ is a two-weight linear code with parameters $[\frac{q^{m}-1}{q^{m_2}-1}(q^{m_2-1}-1),m_1]$ and its weight enumerator is given by Table I.
\end{thm}

\[ \small\begin{tabular} {c} Table I. Weight distribution of the code in Theorem 4.2.  \\
\begin{tabular}{c|c}
   \hline
 Weight & Frequency \\
  \hline
        0   &    1\\
  $\frac{q^{m_1-1}(q-1)(q^{m}-1)(q^{m_2-1}-1-(-1)^{\frac{m_1+m_2}{2}}(q-1)q^{\frac{m_2-m_1}{2}})}{(q^{m_1}-1)(q^{m_2}-1)}$  & $\frac{q^{m_1}-1}{q+1}$ \\
  $\frac{q^{m_1-1}(q-1)(q^{m}-1)(q^{m_2-1}-1-(-1)^{\frac{m_1+m_2}{2}+1}(q-1)q^{\frac{m_2-m_1}{2}-1})}{(q^{m_1}-1)(q^{m_2}-1)}$  & $\frac{q(q^{m_1}-1)}{q+1}$ \\
   \hline
\end{tabular}
\end{tabular}
\]

\begin{proof}
For $e=2$, the weight distributions of $\mathcal{C}_{D}$ can be obtained by Lemma 3.5 and Equation (4.3). It is easy to verify that $w_{H}(\textbf{c}_{b})>0$ for all $b\in \Bbb F_{q^{m_1}}^{*}$ if $(m_1,m_2)\neq (2,2)$, then the dimension equals $m_1$.
\end{proof}

\begin{exa}
Let $m_1=2,m_2=4,m=4$. If $q=2$, then $\mathcal{C}_{D}$ in Theorem 4.2 is an optimal  $[7,2,4]$ two-weight linear code according to the Griesmer bound and has weight enumerator $1+2z^{4}+z^{6}$. If $q=3$, then $\mathcal{C}_{D}$ in Theorem 4.2 is an almost optimal  $[26,2,18]$ two-weight linear code according to the Griesmer bound and has weight enumerator $1+6z^{18}+2z^{24}$.
\end{exa}

\begin{exa}
Let $m_1=4,m_2=6,m=12$ and $q=2$. Then $\mathcal{C}_{D}$ in Theorem 4.2 is a  $[2015,4,1040]$ two-weight linear code. Its weight enumerator is given by
$1+10z^{1040}+5z^{1144}$.
\end{exa}

\subsection{The case $a=1$}
In the following, we determine the weight distribution of $\mathcal{C}_{D}$ for $a=1$.

Denote $n=|D|=|\{x\in \Bbb F_{q^{m}}^{*}:\Tr_{q^{m_{2}}/q}(\mathrm{N}_{q^{m}/q^{m_{2}}}(x))+1=0\}|$. It is clear that
\begin{eqnarray}n=|\ker(\mathrm{N}_{q^{m}/q^{m_{2}}})|\cdot |\ker(\Tr_{q^{m_{2}}/q})|=\frac{q^{m_2-1}(q^{m}-1)}{q^{m_2}-1}.\end{eqnarray}
For $b\in \Bbb F_{q^{m_1}}^{*}$, we denote
$$N(b)=|\{x\in \Bbb F_{q^{m}}^{*}:\Tr_{q^{m_{2}}/q}(\mathrm{N}_{q^{m}/q^{m_{2}}}(x))+1=0\mbox{ and }\Tr_{q^{m_1}/q}(b\mathrm{N}_{q^{m}/q^{m_{1}}}(x))=0\}|.$$
By the basic facts of additive characters, we have that
\begin{eqnarray*}
N(b)&=&\frac{1}{q^{2}}\sum_{x\in \Bbb F_{q^{m}}^{*}}\sum_{y,z\in \Bbb F_{q}}\chi(y\Tr_{q^{m_1}/q}(b\mathrm{N}_{q^{m}/q^{m_{1}}}(x)))\chi(z\Tr_{q^{m_{2}}/q}(\mathrm{N}_{q^{m}/q^{m_{2}}}(x))+z)\\
&=&\frac{1}{q^{2}}\sum_{x\in \Bbb F_{q^{m}}^{*}}\sum_{y,z\in \Bbb F_{q}}\chi_{1}(ybx^{\frac{q^{m}-1}{q^{m_1}-1}})\chi_2(zx^{\frac{q^{m}-1}{q^{m_2}-1}})\chi(z)\\
&=&\frac{q^{m}-1}{q^{2}}+\frac{1}{q^{2}}\sum_{x\in \Bbb F_{q^{m}}^{*}}\sum_{y\in \Bbb F_{q}^{*}}\chi_{1}(ybx^{\frac{q^{m}-1}{q^{m_1}-1}})+\frac{1}{q^{2}}\sum_{x\in \Bbb F_{q^{m}}^{*}}\sum_{z\in \Bbb F_{q}^{*}}\chi_2(zx^{\frac{q^{m}-1}{q^{m_2}-1}})\chi(z)\\
& &+\frac{1}{q^{2}}\sum_{x\in \Bbb F_{q^{m}}^{*}}\sum_{y,z\in \Bbb F_{q}^{*}}\chi_{1}(ybx^{\frac{q^{m}-1}{q^{m_1}-1}})\chi_2(zx^{\frac{q^{m}-1}{q^{m_2}-1}})\chi(z)\\
&=&\frac{q^{m}-1}{q^{2}}+\frac{1}{q^{2}}\sum_{x\in \Bbb F_{q^{m}}^{*}}\sum_{y\in \Bbb F_{q}^{*}}\chi_{1}(ybx^{\frac{q^{m}-1}{q^{m_1}-1}})+\frac{1}{q^{2}}\sum_{x\in \Bbb F_{q^{m}}^{*}}\sum_{z\in \Bbb F_{q}^{*}}\chi_2(zx^{\frac{q^{m}-1}{q^{m_2}-1}})\chi(z)+\frac{\Omega(b)}{q^{2}}.
\end{eqnarray*}
Note that \begin{eqnarray*}\sum_{x\in \Bbb F_{q^{m}}^{*}}\sum_{z\in \Bbb F_{q}^{*}}\chi_2(zx^{\frac{q^{m}-1}{q^{m_2}-1}})\chi(z)&=&\frac{q^{m}-1}{q^{m_2}-1}\sum_{x\in \Bbb F_{q^{m}}^{*}}\sum_{z\in \Bbb F_{q}^{*}}\chi_2(zx)\chi(z)\\
&=&\frac{q^{m}-1}{q^{m_2}-1}\sum_{z\in \Bbb F_{q}^{*}}\chi(z)\sum_{x\in \Bbb F_{q^{m}}^{*}}\chi_2(zx)=\frac{q^{m}-1}{q^{m_2}-1}.
\end{eqnarray*}
From Section 4.1 above, we have
\begin{eqnarray*}\sum_{x\in \Bbb F_{q^{m}}^{*}}\sum_{y\in \Bbb F_{q}^{*}}\chi_{1}(ybx^{\frac{q^{m}-1}{q^{m_1}-1}})=\frac{(q-1)(1-q^{m})}{q^{m_1}-1}.
\end{eqnarray*}
From the discussions above, we obtain that
\begin{eqnarray}
N(b)=\frac{q^{m}-1}{q^{2}}(1+\frac{1}{q^{m_2}-1}-\frac{q-1}{q^{m_1}-1})+\frac{1}{q^{2}}\Omega(b).
\end{eqnarray}

For any $b\in \Bbb F_{q^{m_1}}^{*}$, the weight of a codeword
$$\textbf{c}(b)=(\Tr_{q^{m_{1}}/q}(b\mathrm{N}_{q^{m}/q^{m_{1}}}(d_{1})),\cdots,\Tr_{q^{m_{1}}/q}(b\mathrm{N}_{q^{m}/q^{m_{1}}}(d_{n})))$$
equals
\begin{eqnarray}w_{H}(\textbf{c}(b))=n-N(b)=\frac{(q-1)(q^{m}-1)(q^{m_1+m_2}-1)}{q^{2}(q^{m_1}-1)(q^{m_2}-1)}-\frac{1}{q^{2}}\Omega(b)
\end{eqnarray}
by Equations (4.4) and (4.5). Hence, by Lemma 3.2, the parameters of $\mathcal{C}_{D}$ for $e=l=1$  is
$$[\frac{q^{m_2-1}(q^{m}-1)}{q^{m_2}-1},m_1,\frac{(q-1)(q^{m}-1)q^{m_1+m_2-2}}{(q^{m_1}-1)(q^{m_2}-1)}].$$ Then $\mathcal{C}_{D}$ is an optimal one-weight linear code with respect to the Griesmer bound and is not new as mentioned above. For $e=2$ and $l=1$, the weight distribution of $\mathcal{C}_{D}$ is given in the following.
\begin{thm}
Let $m,m_1,m_2$ be positive integers such that $m_1|m$, $m_2|m$, and $\gcd(\frac{m_1}{e},q-1)=l=1$, where $\gcd(m_1,m_2)=e$. Let $\mathcal{C}_{D}$ be the linear code defined in Equation (1.1) for $a=1$. If $e=2$, then $\mathcal{C}_{D}$ is a two-weight linear code with parameters $[\frac{q^{m_2-1}(q^{m}-1)}{q^{m_2}-1},m_1]$ and its weight enumerator is given by Table II.
\end{thm}

\[ \small\begin{tabular} {c} Table II. Weight distribution of the code in Theorem 4.5.  \\
\begin{tabular}{c|c}
   \hline
 Weight & Frequency \\
  \hline
        0   &    1\\
  $\frac{(q-1)(q^{m}-1)(q^{m_1+m_2}+(-1)^{\frac{m_1+m_2}{2}}q^{\frac{m_1+m_2}{2}+1})}{q^{2}(q^{m_1}-1)(q^{m_2}-1)}$  & $\frac{q^{m_1}-1}{q+1}$ \\
  $\frac{(q-1)(q^{m}-1)(q^{m_1+m_2}+(-1)^{\frac{m_1+m_2}{2}+1}q^{\frac{m_1+m_2}{2}})}{q^{2}(q^{m_1}-1)(q^{m_2}-1)}$  & $\frac{q(q^{m_1}-1)}{q+1}$ \\
   \hline
\end{tabular}
\end{tabular}
\]

\begin{proof}
For $e=2$, the weight distributions of $\mathcal{C}_{D}$ can be obtained by Lemma 3.2 and Equation (4.6). Note that $w_{H}(\textbf{c}_{b})>0$ for all $b\in \Bbb F_{q^{m_1}}^{*}$. Then the dimension equals $m_1$.
\end{proof}

\begin{exa}
Let $m_1=2,m_2=4,m=4$. If $q=2$, then $\mathcal{C}_{D}$ in Theorem 4.5 is an almost optimal $[8,2,4]$ linear code according to the Griesmer bound and has weight enumerator $1+z^{4}+2z^{6}$. If $q=3$, then $\mathcal{C}_{D}$ in Theorem 4.5 is an nearly optimal $[27,2,18]$ linear code, while the corresponding optimal linear codes have parameters $[27,2,20]$.
\end{exa}

\begin{exa}
Let $m_1=4,m_2=6,m=12$ and $q=2$. Then $\mathcal{C}_{D}$ in Theorem 4.5 is a  $[2080,4,1040]$ two-weight linear code. Its weight enumerator is given by
$1+5z^{1040}+10z^{1144}$.
\end{exa}

\begin{thm}
Let $m,m_1,m_2$ be positive integers such that $m_1|m$, $m_2|m$, and $\gcd(\frac{m_1}{e},q-1)=l=2$, where $\gcd(m_1,m_2)=e$. Let $\mathcal{C}_{D}$ be the linear code defined in Equation (1.1) for $a=1$. If $e=1$, then $\mathcal{C}_{D}$ is a two-weight linear code with parameters
$$[\frac{q^{m_2-1}(q^{m}-1)}{q^{m_2}-1},m_1,\frac{(q-1)(q^{m}-1)(q^{m_1+m_2}-q^{\frac{m_1+m_2+1}{2}})}{q^{2}(q^{m_1}-1)(q^{m_2}-1)}]$$ and its weight distribution is given in Table III.
\end{thm}

\[ \small\begin{tabular} {c} Table III. Weight distribution of the code in Theorem 4.8.  \\
\begin{tabular}{c|c}
   \hline
 Weight & Frequency \\
  \hline
        0   &    1\\
  $\frac{(q-1)(q^{m}-1)(q^{m_1+m_2}-q^{\frac{m_1+m_2+1}{2}})}{q^{2}(q^{m_1}-1)(q^{m_2}-1)}$  & $\frac{q^{m_1}-1}{2}$ \\
  $\frac{(q-1)(q^{m}-1)(q^{m_1+m_2}+q^{\frac{m_1+m_2+1}{2}})}{q^{2}(q^{m_1}-1)(q^{m_2}-1)}$  & $\frac{q^{m_1}-1}{2}$ \\
   \hline
\end{tabular}
\end{tabular}
\]
\begin{proof}
The proof is completed by Lemma  3.3 and Equation (4.6).
\end{proof}

\begin{exa}
Let $m_1=2,m_2=3,m=6$ and $q=3$. Then $\mathcal{C}_{D}$ in Theorem 4.8 is a  $[252,2,168]$ two-weight linear code. Its weight enumerator is given by
$1+4z^{168}+4z^{210}$. Its dual is a near-MDS code with parameters $[252,250,2]$.
\end{exa}

\subsection{Shortened linear codes of $\mathcal{C}_{D}$}
It is observed that the weights of the code in Theorems 4.2 have a common divisor $q-1$. This indicates that the code $\mathcal{C}_{D}$ may be punctured into a shorter one.

Assume that $a=0$. Note that $x\in D$ implies that $ux\in D$ for any $u\in \Bbb F_{q}$. Hence, the defining set of $\mathcal{C}_{D}$ in Equation (1.1) can be expressed as
\begin{eqnarray}D=(\Bbb F_{q}^{*})D_1=\{uv:u\in \Bbb F_{q}^{*}\mbox{ and }v\in D_1\},\end{eqnarray}
where $d_i/d_j\not\in \Bbb F_{q}^{*}$ for every pair of distinct elements $d_i,d_j$ in $D_1$. Then we obtain a shortened linear code $\mathcal{C}_{D_1}$ of $\mathcal{C}_{D}$. By Theorem 4.2, we directly obtain the following result.
\begin{cor}
Let $m,m_1,m_2$ be positive integers such that $m_1|m$, $m_2|m$, and $m_2>1$. Denote $\gcd(m_1,m_2)=e$. Let $\mathcal{C}_{D_1}$ be the linear code and its defining set is given in Equation (4.7) for $a=0$. If $e=2$ and $(m_1,m_2)\neq (2,2)$, then $\mathcal{C}_{D_1}$ is a two-weight linear code with parameters $[\frac{(q^{m}-1)(q^{m_2-1}-1)}{(q^{m_2}-1)(q-1)},m_1]$ and its weight enumerator is given by Table IV.
\end{cor}

\[ \small\begin{tabular} {c} Table IV. Weight distribution of the code in Corollary 4.10.  \\
\begin{tabular}{c|c}
   \hline
 Weight & Frequency \\
  \hline
        0   &    1\\
  $\frac{q^{m_1-1}(q^{m}-1)(q^{m_2-1}-1-(-1)^{\frac{m_1+m_2}{2}}(q-1)q^{\frac{m_2-m_1}{2}})}{(q^{m_1}-1)(q^{m_2}-1)}$  & $\frac{q^{m_1}-1}{q+1}$ \\
  $\frac{q^{m_1-1}(q^{m}-1)(q^{m_2-1}-1-(-1)^{\frac{m_1+m_2}{2}+1}(q-1)q^{\frac{m_2-m_1}{2}-1})}{(q^{m_1}-1)(q^{m_2}-1)}$  & $\frac{q(q^{m_1}-1)}{q+1}$ \\
   \hline
\end{tabular}
\end{tabular}
\]

\begin{exa}
 Let $m_1=2,m_2=4,m=4$. If $q=3$, then $\mathcal{C}_{D_{1}}$ in Corollary 4.10 is an optimal  $[13,2,9]$ two-weight linear code according to the Griesmer bound and has weight enumerator $1+6z^{9}+2z^{12}$. Its dual has parameters $[13,11,2]$ which is optimal according to \cite{G}.
\end{exa}

\section{Applications}
In this section, we apply our linear codes to construct secret sharing schemes and strongly regular graphs. We denote by $\mathcal{C}^{\bot}$ the dual code of a code $\mathcal{C}$.
\subsection{Secret sharing schemes from linear codes}
Secret sharing schemes were introduced by Shamir and Blakley for the first time in 1979 \cite{B, S}.  Secret sharing schemes are used in banking
systems, cryptographic protocols, electronic voting systems, and the control of nuclear weapons.

It was shown in \cite{A, YD} that any linear code over $\Bbb F_{q}$ can be employed to construct secret sharing schemes. In order to describe the secret sharing scheme of a linear code (see \cite{AB, DD}), we need to introduce the covering problem of linear codes. The support of a vector $\textbf{c}=\{c_0,c_1,\ldots,c_{n-1}\}\in \Bbb F_{q}^{n}$ is defined as $\{0\leq i \leq n-1:c_i\neq 0\}$. A codeword $\textbf{c}_{1}$ covers a codeword $\textbf{c}_{2}$ if the support of $\textbf{c}_{1}$ contains that of $\textbf{c}_{2}$. A minimal codeword of a linear code $\mathcal{C}$ is a nonzero
codeword that does not cover any other nonzero codeword of $\mathcal{C}$. The covering problem of a linear code is to determine all the minimal codewords of $\mathcal{C}$. From \cite[Theorem 12]{DD}, we know that secret sharing scheme with interesting access structure can be derived from $\mathcal{C}^{\bot}$ provided that each nonzero codeword of a linear code $\mathcal{C}$ is minimal.

If the weights of a linear code $\mathcal{C}$ are close enough to each other, then all nonzero codewords of $\mathcal{C}$ are minimal, as described as follows.
\begin{lem}\cite{AB}
Let $w_{\min}$ and $w_{\max}$ denote the minimum and maximum nonzero Hamming weights of a $q$-ary linear code $\mathcal{C}$, respectively. If $w_{\min}/w_{\max}>\frac{q-1}{q}$, then every nonzero codeword of $\mathcal{C}$ is minimal.
\end{lem}

For the codes in Theorem 4.2 and Corollary 4.10, we have
$$\frac{w_{\min}}{w_{\max}}=\frac{q^{m_2-1}-1-(q-1)q^{\frac{m_2-m_1}{2}}}
{q^{m_2-1}-1+(q-1)q^{\frac{m_2-m_1}{2}-1}}>\frac{q-1}{q}$$ if $m_1+m_2\equiv 0\pmod{4}$, and
$$\frac{w_{\min}}{w_{\max}}=\frac{q^{m_2-1}-1-(q-1)q^{\frac{m_2-m_1}{2}-1}}{q^{m_2-1}-1+(q-1)q^{\frac{m_2-m_1}{2}}}>\frac{q-1}{q}
$$ if $m_1+m_2\equiv 2\pmod{4}$.

For the code in Theorem 4.5, we have
$$\frac{w_{\min}}{w_{\max}}=\frac{q^{m_1+m_2}-q^{\frac{m_1+m_2}{2}}}{q^{m_1+m_2}+q^{\frac{m_1+m_2}{2}+1}}
>\frac{q-1}{q}$$ if $m_1+m_2\equiv 0\pmod{4}$, and
$$\frac{w_{\min}}{w_{\max}}=\frac{q^{m_1+m_2}-q^{\frac{m_1+m_2}{2}+1}}
{q^{m_1+m_2}+q^{\frac{m_1+m_2}{2}}}>\frac{q-1}{q}
$$ if $m_1+m_2\equiv 2\pmod{4}$.

For the code in Theorem 4.8, we have
$$\frac{w_{\min}}{w_{\max}}=\frac{q^{m_1+m_2}-q^{\frac{m_1+m_2+1}{2}}}{q^{m_1+m_2}+q^{\frac{m_1+m_2+1}{2}}}=1-\frac{2}{1+q^{\frac{m_1+m_2-1}{2}}}>\frac{q-1}{q}$$
if $m_1+m_2>3$.

From the discussions above, the linear codes obtained in this paper can be used to construct secret sharing schemes with interesting access structures using the framework in \cite{YD}.
\subsection{Strongly regular graphs from linear codes}
A connected graph with $N$ vertices is called a strongly regular graph with parameters $(N,K,\lambda,\mu)$ if it is regular of valency $K$ and the number of vertices joined to two given vertices is $\lambda$ or $\mu$ according as the two given vertices are adjacent or non-adjacent. The theory of strongly regular graphs was introduced by Bose in 1963 for the first time \cite{BO}.

A code $\mathcal{C}$ is said to be projective if the minimum distance of its dual code $\mathcal{C}^{\bot}$ is at least 3.

The following lemma gives a connection between projective two-weight linear codes and strongly regular graphs.
\begin{lem}\cite{C}
If $\mathcal{C}$ is a projective two-weight $[n,k]$ linear code over $\Bbb F_{q}$ with two nonzero weights $w_1,w_2$, then it is equivalent to a strongly regular graph with the following parameters:
\begin{eqnarray*}
N&=&q^{k},\\
K&=&n(q-1),\\
\lambda &=&K^{2}+3K-q(w_1+w_2)-Kq(w_1+w_2)+q^{2}w_1w_2,\\
\mu &=&\frac{q^{2}w_1w_2}{q^{k}}.
\end{eqnarray*}
\end{lem}

Due to Lemma 5.2, new projective two-weight linear codes yield new strongly regular graphs. Examples in Section 4 show that our codes are not always projective. In particular, we find two classes of projective two-weight codes in the following.
\begin{lem}
Let $m=m_1,m_2=2$ and other notations be the same as those of Theorem 4.5. Then the linear code $\mathcal{C}_{D}$ in Theorem 4.5 is a projective two-weight $[\frac{q(q^{m}-1)}{q^{2}-1},m]$ linear code with weight enumerator
$$1+\frac{q^{m}-1}{q+1}z^{\frac{q^{m}-(-1)^{\frac{m}{2}}q^{\frac{m}{2}}}{q+1}}+\frac{q(q^{m}-1)}{q+1}z^{\frac{q^{m}+(-1)^{\frac{m}{2}}q^{\frac{m}{2}-1}}{q+1}}.$$
\end{lem}
\begin{proof}
The weight enumerator can be directly obtained by Theorem 4.5. We now prove that $\mathcal{C}_{D}$ is projective. Let $A_i,B_i$ denote the numbers of codewords with Hamming weight $i$ in $\mathcal{C}_{D}$ and $\mathcal{C}_{D}^{\bot}$, respectively. Denote
$$w_1=\frac{q^{m}-(-1)^{\frac{m}{2}}q^{\frac{m}{2}}}{q+1},w_2=\frac{q^{m}+(-1)^{\frac{m}{2}}q^{\frac{m}{2}-1}}{q+1},A_{w_1}=\frac{q^{m}-1}{q+1},
A_{w_2}=\frac{q(q^{m}-1)}{q+1}.$$
By the first three Pless Power Moments (see \cite{HP}), we have
\begin{eqnarray*}\left\{
\begin{array}{l}
A_{w_1}+A_{w_2}=q^{m}-1,\\
w_1A_{w_1}+w_2A_{w_2}=(n(q-1)-B_1)q^{m-1},\\
w_1^{2}A_{w_1}+w_{2}^{2}A_{w_2}=\left(n(q-1)\left(n(q-1)+1\right)-B_1\left(q+2(n-1)(q-1)\right)+2B_2\right)q^{m-2}.
\end{array} \right.
\end{eqnarray*}
Note that $n=\frac{q(q^{m}-1)}{q^{2}-1}$. Solving the above system, we have $B_1=B_2=0$. Hence, the minimum distance of $\mathcal{C}_{D}^{\bot}$ is at least 3. The proof is completed.
\end{proof}
\begin{lem}
Let $m=m_1,m_2=2$ and other notations be the same as those of Corollary 4.10. Then the linear code $\mathcal{C}_{D_{1}}$ in Corollary 4.10 is a projective two-weight $[\frac{q^{m}-1}{q^{2}-1},m]$ linear code with weight enumerator
$$1+\frac{q^{m}-1}{q+1}z^{\frac{q^{m-1}+(-1)^{\frac{m}{2}}q^{\frac{m}{2}}}{q+1}}+\frac{q(q^{m}-1)}{q+1}z^{\frac{q^{m-1}-(-1)^{\frac{m}{2}}
q^{\frac{m}{2}-1}}{q+1}}.$$
\end{lem}
\begin{proof}
The weight enumerator can be directly obtained by Corollary 4.10. We now prove that $\mathcal{C}_{D_{1}}$ is projective. Let $A_i,B_i$ denote the numbers of codewords with Hamming weight $i$ in $\mathcal{C}_{D_1}$ and $\mathcal{C}_{D_1}^{\bot}$, respectively. Denote
$$w_1=\frac{q^{m-1}+(-1)^{\frac{m}{2}}q^{\frac{m}{2}}}{q+1},w_2=\frac{q^{m-1}-(-1)^{\frac{m}{2}}
q^{\frac{m}{2}-1}}{q+1},A_{w_1}=\frac{q^{m}-1}{q+1},
A_{w_2}=\frac{q(q^{m}-1)}{q+1}.$$
By the first three Pless Power Moments (see \cite{HP}), we have
\begin{eqnarray*}\left\{
\begin{array}{l}
A_{w_1}+A_{w_2}=q^{m}-1,\\
w_1A_{w_1}+w_2A_{w_2}=(n(q-1)-B_1)q^{m-1},\\
w_1^{2}A_{w_1}+w_{2}^{2}A_{w_2}=\left(n(q-1)\left(n(q-1)+1\right)-B_1\left(q+2(n-1)(q-1)\right)+2B_2\right)q^{m-2}.
\end{array} \right.
\end{eqnarray*}
Note that $n=\frac{q^{m}-1}{q^{2}-1}$. Solving the above system, we have $B_1=B_2=0$. Hence, the minimum distance of $\mathcal{C}_{D_1}^{\bot}$ is at least 3. The proof is completed.
\end{proof}

Lemmas 5.2 and 5.3 yield the following theorem.
\begin{thm}
Let $\gcd(\frac{m}{2},q-1)=1$ and $2|m$. Then there exists a strongly regular graph with the following parameters:
\begin{eqnarray*}
N&=&q^{m},\\
K&=&\frac{q(q^{m}-1)}{q+1},\\
\lambda &=&\frac{q^{m+2}-2q^{2}-3q-(-1)^{\frac{m}{2}}q^{\frac{m}{2}}(1-q)}{(q+1)^{2}},\\
\mu &=&\frac{q(q^{\frac{m}{2}}-(-1)^{\frac{m}{2}})(q^{\frac{m}{2}+1}+(-1)^{\frac{m}{2}})}{(q+1)^{2}}.
\end{eqnarray*}\end{thm}

The following theorem can be directly obtained by Lemmas 5.2 and 5.4.
\begin{thm}
Let $2|m$ and $m\geq 4$. Then there exists a strongly regular graph with the following parameters:
\begin{eqnarray*}
N&=&q^{m},\\
K&=&\frac{q^{m}-1}{q+1},\\
\lambda &=&\frac{q^{m}-3q-2-(-1)^{\frac{m}{2}}q^{\frac{m}{2}+1}(q-1)}{(q+1)^{2}},\\
\mu &=&\frac{q(q^{\frac{m}{2}}-(-1)^{\frac{m}{2}})(q^{\frac{m}{2}-1}+(-1)^{\frac{m}{2}})}{(q+1)^{2}}.
\end{eqnarray*}\end{thm}

We remark that the parameters of the strongly regular graphs in Theorems 5.5 and 5.6 are probably new after comparing with known ones in the literature.

\section{Concluding remarks}
In this paper, we presented a construction of $q$-ary linear codes and determined the weight distributions in some cases based on Gauss sums. Four classes of two-weight linear codes were obtained. Note that these linear codes have very flexible parameters and are probably new after comparing with known two-weight linear codes in the literature (see \cite{C, D, DD, DD2, HY0, ZF} for some known two-weight linear codes). It is interesting that our construction can produce optimal or almost optimal codes. What's more, our codes can be used to construct secret sharing schemes with interesting access structures and strongly regular graphs.
\section*{Acknowledgments}
The authors are very grateful to the reviewers and the Editor for their valuable comments that improved the quality of this paper. Special thanks go to one of the reviewers for pointing out some knowledge of one-weight linear codes.

\end{document}